\documentclass[1p,number,11pt,compress]{elsarticle}
\usepackage{latexsym,amsmath,amsfonts,amssymb,graphics,color,bm,mathrsfs}
\usepackage{amsthm} 
\journal{Journal of Differential Equations}

\usepackage{xcolor}              
\definecolor{a1}{rgb}{0,0,0.8}   
\usepackage{colortbl}	           
\usepackage[colorlinks=true,linktocpage=true,hyperindex=true,pageanchor=true,bookmarks=true]{hyperref} 
\hypersetup{citecolor=a1,linkcolor=a1,menucolor=a1,urlcolor=a1,runcolor=a1}
\hypersetup{citebordercolor=a1,linkbordercolor=a1,urlbordercolor=a1,runbordercolor=a1}
\hypersetup{hypertexnames=false} 

\newcommand{\be}{\begin{equation}}
\newcommand{\ee}{\end{equation}}

\newtheorem{thm}{Theorem}
\newtheorem{lem}[thm]{Lemma}

\begin{document}
\begin{frontmatter}
\title{Equations of Motion for Variational Electrodynamics}
\author[fsc]{Jayme De Luca}
\ead{jayme.deluca@gmail.com}
\address[fsc]{Departamento de F\'{\i}sica,
Universidade Federal de S\~{a}o Carlos,
S\~{a}o Carlos, S\~{a}o Paulo 13565-905, Brazil}

\begin{abstract}

We extend the variational problem of Wheeler-Feynman electrodynamics by generalizing the electromagnetic functional to a local space of absolutely continuous trajectories possessing a derivative (velocities) of bounded variation. We show here that the Gateaux derivative of the generalized functional defines \emph{two} partial Lagrangians for variations in our generalized local space, one for each particle.  We prove that the critical-point conditions of the generalized variational problem  are: (i) the Euler-Lagrange equations must hold Lebesgue-almost-everywhere and (ii)  the momentum of each partial Lagrangian and the Legendre transform of each partial Lagrangian must be absolutely continuous functions, generalizing the Weierstrass-Erdmann conditions.
\end{abstract}


\begin{keyword}
 calculus of variations, absolute continuity, neutral differential-delay equations, state-dependent delay.
\end{keyword}

\end{frontmatter}

\section{Introduction}
\subsection{Significance of the variational formulation}

Electrodynamics has \emph{neutral differential-delay equations} of mixed type with implicitly defined \emph{state-dependent} delays for the motion of point charges\cite{WheeFey}, which theory is still a challenge for present day mathematics. The theory of differential-delay equations with state-dependent delay initiated in the 70's with the foundations based on infinite-dimensional dynamical systems \cite{Mallet-Paret1,Mallet-Paret2,JackHale,Hans-Otto, Mallet-Paret3, HKWW06, Nicola1, Angelov} and the numerical studies \cite{Nicola1,Bellman,Driverbook,Nicola2} (see also Ref. \cite{BellenZennaro} for an extensive list of references).
A formal variational structure for the electromagnetic equations is known since 1903\cite{WheeFey,Schwarz}, but only recently the variational structure has been embedded into a variational principle\cite{JMP2009,minimizer,cinderela}. The existence of the variational principle is important for the analytic and numerical studies because one is dealing with functional minimization  \cite{JMP2009,CAM}, which makes the electromagnetic equations special in the class of neutral differential-delay equations.
\par
\subsection{What is this paper about}
Here we invert the direction of application of the variational principle of Refs. \cite{JMP2009,minimizer,cinderela} by extending the local domain of trajectory variations to the set of absolutely continuous orbits possessing velocities of bounded variation. The generalized electromagnetic variational problem is studied for variations belonging to the local normed space $X_{BV}$ of \textit{absolutely continuous} orbits possessing a velocity of \textit{bounded variation}. 
\par
The classical problem of the calculus of variations studies a functional of the classical mechanical form, $F \equiv \int_{0}^{T} \mathcal{L} (\mathbf{x},\dot{\mathbf{x}}, t)dt $ \cite{Gelfand,Troutman}, usually minimized on a domain of continuous and piecewise $C^2$ orbits possessing velocity discontinuities on a finite grid of times  (henceforth breaking points). The critical-point conditions of the former functional are (i) at the breaking points the momentum $P \equiv \partial \mathcal{L} / \partial \dot{\mathbf x}$ and the Legendre transform $E (\mathbf{x},\dot{\mathbf x}, t) $ of the Lagrangian $ \mathcal{L} (\mathbf{x},\dot{\mathbf x}, t) $, defined as $E (\mathbf{x},\dot{\mathbf x}, t) \equiv (\dot{\mathbf{x}}\cdot \partial \mathcal{L}/\partial \dot {\mathbf{x}} )- \mathcal{L} (\mathbf{x},\dot{\mathbf x}, t)$,   must be \textit{continuous} functions (henceforth the Weierstrass-Erdmann corner conditions\cite{Gelfand, Troutman}) and (ii) the Euler-Lagrange equation should hold on all other points \cite{Gelfand, Troutman}. 
  \par
The electromagnetic functional does \textit{not} have the classical mechanical form $F \equiv \int_{0}^{T} \mathcal{L} (\mathbf{x},\dot{\mathbf{x}}, t)dt $ \cite{Gelfand}, and despite the importance to electrodynamics its critical-point conditions have not been studied in functional analytic detail.  While the generalized electromagnetic functional is {\it{not}} of the classical mechanical type, we show here that its first variation (the Gateaux derivative) decomposes into a \textit{sum} involving \textit{two} partial Lagrangians of the former type on our extended local space $X_{BV}$,  i.e., $\delta S=\delta  S_1 + \delta S_2$ with $S_i \equiv \int_{0}^{T} \mathcal{L}_i (\mathbf{x},\dot{\mathbf x}, t)dt $ \, for $i=1,2$.

\par
After generalizing the domain of the electromagnetic functional to $X_{BV}$ we prove here that the generalized critical-point conditions are (i) the momentum of each partial Lagrangian and the Legendre transform of each partial Lagrangian must be \emph{absolutely continuous} functions and (ii) the Euler-Lagrange equations must be satisfied Lebesgue-almost-everywhere, which is a well-defined request because velocities of bounded variation have a derivative Lebesgue-almost-everywhere. 
\par
References \cite{minimizer, cinderela} studied the variational two-body problem in a domain $X_{\widehat{C}^2} $ of continuous and piecewise $C^2$ orbits possessing discontinuous velocities on a finite grid of times. The critical-point conditions of Refs. \cite{minimizer,cinderela} are Euler-Lagrange equations holding piecewise \emph{and} the Weierstrass-Erdmann corner conditions\cite{Gelfand} that the momenta and the partial energies are continuous. The absolute continuity condition is not part of the results of \cite{minimizer, cinderela} and is a stronger version of Weierstrass-Erdmann conditions coming from our extension of the electromagnetic domain to $X_{BV}$.

\subsection{Existence and uniqueness results}
The neutral differential-delay equations with state-dependent delay\cite{WheeFey, JMP2009} connected to the electromagnetic variational principle are still not well understood in terms of the nature of solutions and their existence and uniqueness. The early studies found a one-parameter family of $C^\infty$ circular-orbit-solutions for the two-body problem \cite{Schoenberg, Schild}. Some $C^\infty$ circular-orbit-solutions for more than two charges are discussed in Ref. \cite{Martinez}. The earliest discussion of admissible solutions for similar equations is found in Refs. \cite{Driver_A, Driver_B}.  In Refs. \cite{Driver_C, Driver_D} an existence result for $C^{\infty}$ solutions of the equal charges problem (repulsive interaction) with initial condition restricted to a line was proved with the contraction mapping principle. The contraction mapping of Refs.  \cite{Driver_C, Driver_D} holds only for globally large charge separations and converges to a $C^\infty$ solution. For the repulsive problem orbits and delays become asymptotically unbounded, unlike the globally bounded orbits that are possible for the opposite charges problem\cite{Schild, Martinez}.  The conditions for a well-posed $C^{\infty}$ solution of the general two-body problem are discussed in Ref. \cite{Murdock}. 
\par
In the days of Wheeler and Feynman\cite{WheeFey} the modern  theory of delay equations was not out \cite{HKWW06, JackHale, Hans-Otto}, and the equations were originally studied for $C^\infty$ solutions only. Neutral differential-delay equations can propagate a velocity discontinuity\cite{BellenZennaro}, a property directly related to the existence of extrema with velocity discontinuities for the variational problem. The relation between solutions of the Wheeler-Feynman neutral differential-delay equations and the variational problem is discussed in Ref. \cite{minimizer}. After our generalization of Ref. \cite{JMP2009} to a variational boundary value problem, it became natural to search for solutions with velocity discontinuities on a set of zero measure. The present author and collaborators have studied solutions with velocity discontinuities for the two-body problem in Refs. \cite{cinderela,CAM,Daniel}. In Ref. \cite{cinderela} we gave the first existence and uniqueness result for orbits with discontinuous velocities and boundary data near those of circular-orbits of large radii. The existence and uniqueness result of Ref. \cite{cinderela} is studied in further detail for the simplest type of boundary data in Theorem 2 of Ref. \cite{Daniel}.

\subsection{How this paper is divided}
In Section \ref{functional} we explain the variational boundary-value problem for the electromagnetic two-body system. In Section \ref{sectionXC2} we review a simple derivation of the critical-point condition in $X_{{\widehat{C}}^2}$ to prepare the stage for the extension to a variational problem in $X_{BV}$. We start Section \ref{Stieltjes} proving some lemmas necessary to convert a Lebesgue integral  into a Stieltjes integral and then perform an integration by parts in the expression for the first variation  (Gateaux derivative). In this same Section \ref{Stieltjes} we derive the generalized critical-point conditions. We start from the problem of deriving the Weierstrass-Erdmann  continuity conditions of the partial momenta using a functional defined from the trajectories in $\mathbb{R}^3$ and an integration by parts.
  The Legendre-transform condition is studied in the appendix applying the same integration by parts technique of the main text and exploring a linearity property of the electromagnetic functional to generalize it to a larger domain $T_{BV} \otimes X_{BV} \equiv \mathbb{R} \times \mathbb{R}^3$ containing $X_{BV}$.  In the enlarged space $T_{BV} \otimes X_{BV} \equiv \mathbb{R} \times \mathbb{R}^3$ the time coordinate of each particle is a monotonically increasing real function of the independent variable possessing a derivative of bounded variation.  The Legendre transforms of the partial Lagrangians appear as the fourth components of the momenta of the variational problem for the generalized functional in $T_{BV} \otimes X_{BV}$. 

\section{Boundary-value problem}
\label{functional}

 We write the electromagnetic functional \cite{JMP2009} in units where the speed of light is $c \equiv 1$, the electronic charge and electronic mass are $e_{1} \equiv -1$ and $m_{1}$, respectively, and the protonic charge and protonic mass are $e_{2} \equiv 1$ and $m_2$, respectively. We henceforth use the index  $i=1$ to denote the electronic trajectory and $i=2$ to denote the protonic trajectory. Each absolutely continuous trajectory of $X_{BV}$ is a real function of time, $t \rightarrow  \mathbf{x}_i(t) \in \mathbb{R}^3 $, possessing a derivative ${\dot{\mathbf{x}}}_i(t)$ of bounded variation. Central to the construction of the electromagnetic functional are the light-cone conditions 
\begin{equation}
t_j^{\pm}= t \pm \vert {\mathbf{x}_i}(t)-{\mathbf{x}}_{j}(t_j^{\pm}) \vert \equiv t \pm r_{ij}^{\pm},
\label{light-cone}
\end{equation}
where 
\begin{eqnarray}
r_{ij}^{\pm} \equiv \vert {{\mathbf{x}_i(t) -\mathbf{x}}_{j}(t \pm r_{ij}^{\pm})} \vert , 
\label{defrij}
\end{eqnarray}
is the Euclidean norm of the spatial separation in light-cone and $j \equiv 3-i$ for $i=1,2$.
 Equation (\ref{light-cone}) is an implicit condition to be solved for $t_{j}^{\pm}(t)$ with given trajectories $\mathbf{x}_i (t) \in X_{BV}$ and $\mathbf{x}_j (t) \in X_{BV}$,  having a state-dependency on either the advanced or the retarded coordinates $\mathbf{x}_{j}(t_{j}^{\pm}(t) ) $. In Eqs. (\ref{light-cone}) and (\ref{defrij}) the plus sign defines the future light-cone condition and the minus sign defines the past light-cone condition. 
   \par
In order for (\ref{light-cone}) to have unique solutions and in order for the electromagnetic functional to be well-defined it is necessary that both trajectories have a velocity lesser than the speed of light,
\begin{equation}
 \vert \dot{\mathbf{x}}_i(t) \vert < 1,
\label{sub-cond}
\end{equation}
for $i=1,2$ wherever the derivative is defined in $X_{BV}$, henceforth sub-luminal orbits. It is shown in Proposition 1 of Ref. \cite{JMP2009} that sufficiently small neighbourhoods of sub-luminal orbits contain only sub-luminal orbits. Therefore, all orbits on the local space $X_{BV}$ of a sub-luminal orbit are sub-luminal orbits. The retarded and the advanced deviating arguments $t_j^{\pm}(t)$ appearing everywhere in the electromagnetic problem are a manifestation of the Einstein locality condition, which demands that only trajectory points satisfying the light-cone condition (\ref{light-cone}) should interact, i.e., point $(t,\mathbf{x}_i(t))$ with point $(t_j^-(t), \mathbf{x}_j(t_j^{-}(t))$ and point $(t,\mathbf{x}_i(t))$ with point $(t_j^+(t), \mathbf{x}_j(t_j^{+}(t))$.

\par
According to Lemma 1 of Ref. \cite{CAM}, for continuous sub-luminal orbits possessing a derivative satisfying (\ref{sub-cond}) Lebesgue-almost-everywhere, the past and the future light-cone conditions (\ref{light-cone}) define unique maps for the deviating arguments   
\begin{eqnarray}
t_{1}\rightarrow t_{2}^+(t_1, \mathbf{x}_1 (t_{1})), &\; & \;t_{1}\rightarrow \notag
t_{2}^{-}(t_1, \mathbf{x}_1 (t_{1})), \notag \\
t_{2}\rightarrow t_{1}^{+}(t_2, \mathbf{x}_2 (t_{2})), &\; &t_{2}\rightarrow t_{1}^{-}(t_2,\mathbf{x}_2 (t_{2})), \notag 
\end{eqnarray}
which are functions $ t_{j}^{\pm}(t_i, \mathbf{x}_i (t_{i}))$ of the independent variables $(t_i, \mathbf{x}_i(t_i))$ possessing partial derivatives Lebesgue-almost-everywhere defined by the implicit function theorem and Eq. (\ref{light-cone}) as
\begin{eqnarray}
\frac{\partial t_{j}^{\pm}}{\partial \mathbf{x}_i}&=&\frac{\pm \mathbf{n}_{ij}^{\pm}}{(1 \pm \mathbf{n}_{ij}^{\pm} \cdot {\dot{\mathbf{x}}_j^{\pm}} ) }, \label{partialtjpartialxi}\\
\frac{\partial t_{j}^{\pm}}{\partial t_i}&=&\frac{1}{(1 \pm \mathbf{n}_{ij}^{\pm} \cdot {\dot{\mathbf{x}}_j^{\pm}} ) } \label{partialtjpartialti}, 
\end{eqnarray}
where the unit vector 
\begin{equation}
\mathbf{n}_{ij}^{\pm } \equiv  (\mathbf{x}_i -\mathbf{x}_{j}^{\pm}) / \vert {{\mathbf{x}_i -\mathbf{x}}_{j}^{\pm}} \vert  , 
\label{unit-n}
\end{equation}
points from either the advanced or the retarded position $\mathbf{x}_{j}^{\pm} \equiv \mathbf{x}_{j}(t_j^{\pm}(t_i, \mathbf{x}_i(t_i)))$ to the position $\mathbf{x}_i(t)$ for each pair $(i,j)$ with $i=1,2 $ and $j \equiv 3-i$\cite{cinderela}. 
\par
The variational problem is the critical-point-condition for trajectory segments $(O_{1}, L_2^-)$ (blue) and $(O_1^+, L_2)$ (green), which should satisfy the boundary conditions illustrated in FIG. \ref{Squema}, i.e., (a) have the specified initial point $O_1$ for
trajectory $1$ and have the solid boundary-segment illustrated by the red triangle on the left of FIG. \ref{Squema} to agree with the segment of trajectory $2$ inside
the light-cone of point $O_1$, and (b) have the final point $L_2$ for 
trajectory $2$ and have the solid boundary-segment illustrated by the red triangle on the right of FIG. \ref{Squema} to agree with the segment of trajectory $1$ inside the light-cone of point $L_2$\cite{JMP2009,minimizer,cinderela}.

\begin{figure}[h!]
   \centering
    \includegraphics[scale=0.4]{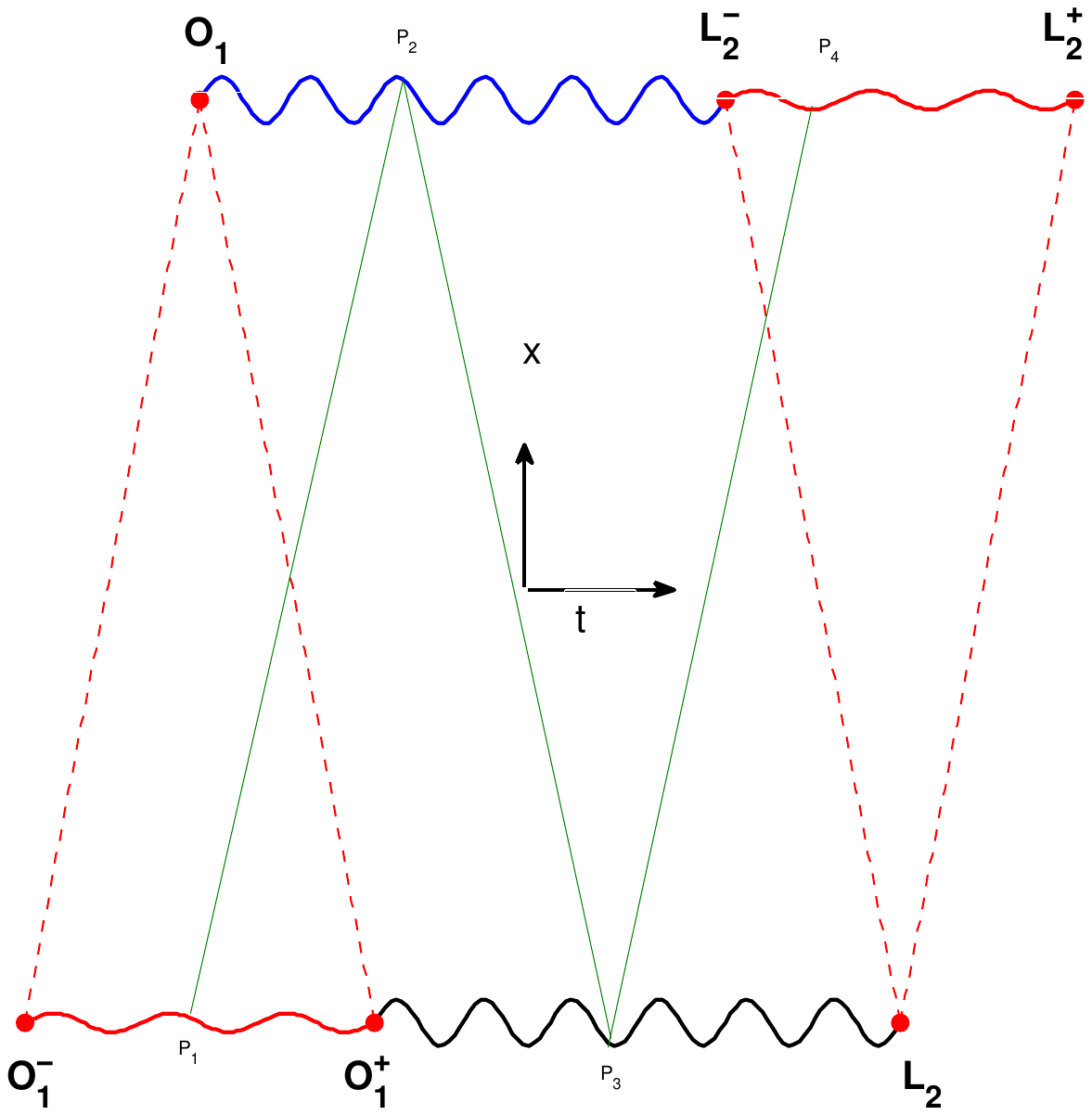} 
   \caption{The boundary conditions in $\mathbb{R}^3$ are (a) the
initial point $O_{1} \equiv \mathbf{x}_1(t_{O_1}) $ of trajectory  $1$
and the trajectory segment of $\mathbf{x}_2(t_2)$ for $ t_2\in [t_{O_1^-},t_{O_1^+}]$ (solid red line) at which endpoints the position $\mathbf{x}_2$ is in the light-cone condition with $O_1$ (indicated by broken red lines), and (b) the endpoint $L_{2} \equiv  \mathbf{x}
_2 (t_{L_2})$ of trajectory $2$ and the respective trajectory segment of $%
\mathbf{x}_1(t_1)$ for $t_1\in[t_{L_2^-},t_{L_2^+}]$ (solid red line), at which endpoints the position $\mathbf{x}_1$ is in the light-cone condition with $L_2$ (also indicated by broken red lines). Illustrated is a sewing chain $P_1,P_2,P_3,P_4$ (solid green lines).Trajectories $\mathbf{x}_1(t_1)$
 for $t_1\in[t_{O_1},t_{L_2^-}]$ (solid blue line) and $\mathbf{x}_2(t_2)$ for $t_2\in[%
t_{O_1^+},t_{L_2}]$ (solid black line) are determined by the critical-point condition.  Arbitrary units.} 
\label{Squema}
\end{figure}
\par
The electromagnetic functional is a sum of four Lebesgue integrals over the particle's times, as defined in the following. Two integrals are \textit{local} integrals involving one trajectory only, i.e., $\int \mathcal{M}_i (\mathbf{x}_i,\dot{\mathbf{x}}_i)dt_i$ for $i=1,2$. The other two Lebesgue integrals are \textit{interaction} integrals depending on both positions and velocities, where one position and velocity is evaluated at a deviating argument, $\int \mathcal{I}^\pm_{ij}({\mathbf{x}}_i, \dot{\mathbf{x}}_i,{\mathbf{x}}_{j}^{\pm},\dot{\mathbf{x}}_{j}^{\pm})dt_i$, for each $(i,j)$ pair with $i=1,2$ and $j=3-i$. By changing the integration variable of the interaction integrals, the electromagnetic functional can be expressed in two equivalent forms
\begin{align} \label{aFokker}
S[\mathbf{x}_1 , \mathbf{x}_2] \equiv &\negthickspace \int_{t_{O_1^{+}}}^{t_{L_{2}}}\negthickspace \mathcal{M}_{2}dt_{2}
+\negthickspace \int_{t_{O_{1}}}^{t_{L_2^{-}}}\negthickspace  \mathcal{M}_{1}dt_{1}+\underbrace{ \int_{t_{O_{1}}}^{t_{L_2^{+}}} \negthickspace \mathcal{I}_{12}^-dt_{1}}
+ \underbrace{\negthickspace \int_{t_{O_1}}^{t_{L_2^-}}\negthickspace  \mathcal{I}_{12}^+dt_{1}},\\
& \qquad \qquad   \qquad\qquad\qquad\qquad \quad \; \, \, \; \;  \;\Updownarrow \qquad\qquad\quad\; \; \Updownarrow \notag\\
\negthickspace \negthickspace =&\negthickspace \int_{t_{O_{1}}}^{t_{L_2^{-}}} \negthickspace \mathcal{M}_{1}dt_{1}
+\negthickspace \int_{t_{O_1^{+}}}^{t_{L_{2}}}\negthickspace \negthickspace \mathcal{M}_{2}dt_{2}
+ \overbrace{ \int_{t_{O_1^-}}^{t_{L_2}} \negthickspace \mathcal{I}_{21}^+dt_{2}}
+ \overbrace{ \int_{t_{O_1^+}}^{t_{L_{2}}} \negthickspace \mathcal{I}_{21}^-dt_{2}}.
\label{L2}
\end{align}
The vertical arrows linking an interaction integral of Eq. (\ref{aFokker}) to another of Eq. (\ref{L2}) indicate equality under a change of the integration variable using the state-dependent condition (\ref{light-cone}). 
The advanced and retarded times $t_j^{\pm}$ are absolutely continuous and monotonically increasing functions of the other particle's time $t_i$, as required for a change of integration variable (e.g. see Ref. \cite{BelaNagy}).  For orbits in $X_{BV}$ the Radon-Nikodym derivative is defined Lebesgue-almost-everywhere by  
\begin{equation}
\frac{d t_{j}^{\pm}}{d t_{i}}=\frac{\partial t_j^{\pm}}{\partial t_i}+ {\dot{\mathbf{x}}_i}\cdot \frac{\partial t_j^{\pm} }{\partial \mathbf{x}_i   }= \frac{(1 \pm \mathbf{n}_{ij}^{\pm} \cdot {\dot{\mathbf{x}}}_i)}{(1 \pm \mathbf{n}_{ij}^{\pm} \cdot {\dot{\mathbf{x}}}_{j}^{\pm})},  \label{t2difft1}
\end{equation}
where we have used Eqs. (\ref{partialtjpartialxi}) and  (\ref{partialtjpartialti}).
 In Eq. (\ref{t2difft1}),  the abbreviations $\dot{\mathbf{x}}_i $ and  $\dot{\mathbf{x}}_j^{\pm} $ denote the velocities evaluated respectively at $t_i$ and at either the advanced or the retarded deviating arguments $t_j^{\pm}(t_i, \mathbf{x}_i)$, and the unit vector $\mathbf{n}_{ij}^{\pm}$ is defined by Eq. (\ref{unit-n}). 
  \par
Here we consider the electromagnetic variational structure defined by constraints (\ref{light-cone}) and functionals (\ref{aFokker}) and (\ref{L2}) with
\begin{eqnarray}
\mathcal{M}_i\equiv m_i ( 1-\sqrt{1-\dot{\mathbf{x}}_i^2 }   \;), \label{MU}
\qquad \\
\mathcal{I}^\pm_{ij}(\mathbf{x}_i,\dot{\mathbf{x}}_i,{\mathbf{x}}_{j}^{\pm},{\dot{\mathbf{x}}}_{j}^{\pm})
\equiv \frac{(1-\dot{\mathbf{x}}_i \cdot \dot{\mathbf{x}}_{j}^{\pm}) }{2r_{ij}^{\pm}( 1 \pm \mathbf{n}_{ij}^{\pm}\cdot \dot{\mathbf{x}}_{j}^{\pm})  },  \label{EU}
\end{eqnarray}
where $r_{ij}^{\pm}$ is given by (\ref{defrij}) and again $j \equiv 3-i$ for $i=1,2$\cite{cinderela}. Notice that along a sub-luminal orbit of $X_{BV}$ Eqs. (\ref{sub-cond}), (\ref{MU}) and (\ref{EU}) yield $\mathcal{M}_i > 0$ and $\mathcal{I}^{\pm}_{ij}  > 0$ almost everywhere, thus defining a semi-bounded functional ($S > 0$) by either (\ref{aFokker}) or (\ref{L2}). 
The integrands of type  (\ref{EU}) include denominators that should be non-zero. For that we restrict here to non-collisional ($r_{ij}^{\pm} >0$) and sub-luminal orbits ($|\dot{\mathbf{x}}_j^{\pm}|<1$ ), a requirement that could be relaxed to non-zero denominators outside sets of zero measure, as discussed in Refs. \cite{Gordon1,Gordon2}. 
Finally, velocities of bounded variation form a Banach algebra\cite{Ambrosio} and therefore (\ref{MU}) and (\ref{EU}) are functions of bounded variation which are locally integrable, thus making the electromagnetic functional (\ref{aFokker}) well-defined along non-collisional sub-luminal orbits.
\par

\section{Variational problem in $X_{\widehat{C}^2} \subset X_{BV} $ }
\label{sectionXC2}

The nuts and bolts to derive the critical-point conditions is an integration by parts explained in the following.  A generic trajectory variation in either local space $X_{\widehat{C}^2} $ or $X_{BV}$ is defined by
\begin{eqnarray}\label{perturb}
\mathbf{x}_{1}^{v}(t) &= & \mathbf{x}_{1}(t)+\epsilon \mathrm{b}_{1}(t)  \qquad\mbox{and} \qquad   \dot{\mathbf{x}}_{1}^{v}(t) = \dot{\mathbf{x}}_{1}(t) +\epsilon \dot{ \mathrm{b}}_{1}(t), \\
\mathbf{x}_{2}^{v}(t) &= & \mathbf{x}_{2}(t)+\epsilon {\mathrm{b}}_{2}(t) \qquad\mbox{and}\qquad   \dot{\mathbf{x}}_{2}^{v}(t) = \dot{ \mathbf{x}}_{2}(t)+\epsilon \dot{\mathrm{b}}_{2}(t). \label{perturb2}
\end{eqnarray}
with the $\mathrm{b}_i(t)$  satisfying Dirichlet boundary conditions
\begin{eqnarray} \label{boundaries}
\mathrm{b}_{1}(t_{O_{1}})& =& 0 \qquad\mbox{and}\qquad \mathrm{b}_{1}(t_{L_{2}^{-}})=0, \\
\mathrm{b}_{2}(t_{O_{1}^{+}} ) & =& 0 \qquad\mbox{and}\qquad \mathrm{b}_{2}(t_{L_{2}})=0, \label{boundaries2}
\end{eqnarray}
where $\epsilon > 0 $ and the upper $v$ denotes a varied trajectory  in either $X_{\widehat{C}^2} $ or $X_{BV}$.

The first variation of the electromagnetic functional (\ref{aFokker}) (Gateaux derivative) is defined  by 
\begin{eqnarray}
\delta S (\mathrm{b}_1, \mathrm{b}_2) \equiv \lim_{\epsilon  \rightarrow 0} \frac{S[{\mathbf{x}_1}+\epsilon \mathrm{b}_1, {\mathbf{x}_2}+\epsilon \mathrm{b}_2 ] -S[\mathbf{x}_1, \mathbf{x}_2] }{\epsilon} . \label{Gateaux}
\end{eqnarray}
The Gateaux derivative (\ref{Gateaux}) naturally splits in a sum of two terms, $\delta S = \delta S_1 + \delta S_2$, as follows: Variation $\delta S_1$ is evaluated by holding trajectory $2$ constant while the absolutely continuous trajectory $1$ is varied.  Variation $\delta S_2$ is evaluated by holding trajectory $1$ constant while the absolutely continuous trajectory $2$ is varied.
 The linear variation $\delta S_1$ is calculated using Eq. (\ref{aFokker}) with its first term kept constant, while the linear variation $\delta S_2$ is obtained in the same manner by varying trajectory $2$ while trajectory $1$ is kept constant and using the equivalent expression (\ref{L2}) for the electromagnetic functional (with its second term kept constant).
The non-constant part of either integrand  (\ref{aFokker}) or  (\ref{L2}) is henceforth called the partial Lagrangian $i$,
\begin{eqnarray}
\mathcal{L}_i(\mathbf{x}_i, \dot{\mathbf{x}}_i,t) \equiv \mathcal{M}_i + \mathcal{I}_{ij}^{-} +\mathcal{I}_{ij}^{+}  \label{defpartial}, 
\end{eqnarray}
for $i=1,2$ and $j \equiv 3-i$. The last $t$ argument on the left-hand-side of Eq. (\ref{defpartial}) is a time dependence brought in from the dependence on trajectory $\mathbf{x}_j$ because
 the light-cone conditions (\ref{light-cone}) define deviating arguments $t_{j\pm}(t, \mathbf{x}_i)$ which depend on $t$ and $\mathbf{x}_i$ for a fixed trajectory $\mathbf{x}_j(t)$. The extra dependence on $\mathbf{x}_i$ brought in by the dependence of $t_{j\pm}$ on $\mathbf{x}_i$ is assumed and abbreviated as explained below Eq. (\ref{t2difft1}). The left-hand-side of Eq. (\ref{defpartial}) is an abbreviation for 
$ \mathcal{L}_i(\mathbf{x}_i, \dot{\mathbf{x}}_i, \mathbf{x}_j ^{+} (t,\mathbf{x}_i) ,  \mathbf{x}_j^{-} (t,\mathbf{x}_i), \dot{\mathbf{x}}_j^{+}(t,\mathbf{x}_i) ,  \dot{\mathbf{x}}_j^{-}(t,\mathbf{x}_i) ) $.

In order to use Lebesgue dominated convergence to exchange the order of the $\epsilon$-limit and the integral in Eq. (\ref{Gateaux}) we need the partial derivatives of $ \mathcal{L}_i(\mathbf{x}_i, \dot{\mathbf{x}}_i,t)$ to exist and be bounded along the sub-luminal orbits of either $X_{{\widehat{C}}^2}$ or $X_{BV}$, which is seen to be the case as follows: \par
(a) The partial derivative respect to $\dot {\mathbf{x}}_i$ is henceforth called the momentum function,
\begin{eqnarray}
P_i(t)  \equiv \frac{\partial \mathcal{L}_i}{\partial \dot{\mathbf{x}}_i}(\mathbf{x}_i(t), \dot{\mathbf{x}}_i(t),t),
\label{P1}
\end{eqnarray} 
which evaluated using Eqs. (\ref{MU}), (\ref{EU}), (\ref{DEFA}) and (\ref{DEFU}) yields
\begin{eqnarray}
P_i (t) = \frac{m_{i} \dot{\mathbf{x}}_{i}  }{ \sqrt{1 - \dot{\mathbf{x}}_i^2}  }   
-\frac{\dot{{\mathbf{x}}}_{j}^{-} }{2r_{ij}^{-}(1-\mathbf{n}_{ij}^{-}\cdot \dot{\mathbf{x}}_{j}^{-})} -
\frac{\dot{{\mathbf{x}}}_{j}^{+}}{2r_{ij}^{+}(1+\mathbf{n}_{ij}^{+}\cdot \dot{\mathbf{x}}_{j}^{+})}.
\label{momentum1} 
\end{eqnarray} 
It can be seen by inspection that the denominators on the right-hand-side of Eq. (\ref{momentum1}) are bounded away from zero in either $X_{\widehat{C}^2}$ or $X_{BV}$ because sub-luminal orbits satisfy condition (\ref{sub-cond}). Moreover, the momentum $P_1(t)$ defined by Eq. (\ref{momentum1}) in $X_{BV}$ is a function of bounded variation along a sub-luminal non-collisional orbit of $X_{BV}$ by the Banach-algebra property of $X_{BV}$. \par
(b) The partial derivative respect to $\mathbf{x}_i$ can be evaluated using Eqs. (\ref{partialtjpartialxi}), (\ref{partialtjpartialti}) and (\ref{aFokker}). To perform the calculation it is convenient to express the denominator of $\mathcal{I}_{ij}^{\pm} $ as
\begin{eqnarray}
D_{j}^{\pm} \equiv r_{ij}^{\pm} (1 \pm \mathbf{n}_{ij}^{\pm} \cdot \dot{\mathbf{x}}_j)=\pm ( t_{j}^{\pm} -t +r_{ij}^{\pm}  \mathbf{n}_{ij}^{\pm} \cdot \dot{\mathbf{x}}_j^{\pm} ), \label{defD}
\end{eqnarray}
by use of Eq. (\ref{light-cone}), where center dot denotes the scalar product of $\mathbb{R}^3$. The partial derivative of (\ref{defD}) with respect to $\mathbf{x}_i$ at a fixed $t_i=t$ is
\begin{eqnarray}
\frac{\partial D_j^{\pm}} {\partial \mathbf{x}_i}=\pm (1 -  \mathbf{v}_{j\pm}^2 + r_{ij}^{\pm} \mathbf{n}_{ij}^{\pm} \cdot \mathbf{a}_{j\pm})\frac{\partial t_j^{\pm} }{\partial \mathbf{x}_j}, \label{deriD}
\end{eqnarray}
where $\mathbf{v}_{j\pm} \equiv \dot{\mathbf{x}}_j (t_{j}^{\pm})$ and  $\mathbf{a}_{j\pm} \equiv \ddot{\mathbf{x}}_j (t_{j}^{\pm})$ and again center dot denotes the scalar product of $\mathbb{R}^3$. Equation (\ref{deriD}) involves the acceleration of particle $j$ evaluated at the time $t_j^{\pm}$, which exists almost everywhere and is Lebesgue measurable by a property of velocities of bounded variation \cite{Leoni}. Using the above we calculate 
\begin{eqnarray}
\frac{\partial \mathcal{I}_{ij}^{\pm}}{\partial \mathbf{x}_i}=-\frac{\mathcal{I}_{ij}^{\pm}}{D_{j}^{\pm}}\frac{\partial D_j^{\pm}} {\partial \mathbf{x}_i} -\frac{\dot{\mathbf{x}}_i \cdot \mathbf{a}_{j\pm}}{2D_{j}^{\pm}}\frac{\partial {t_{j} ^{\pm}}}{\partial \mathbf{x}_i},
\end{eqnarray}
and at last, from Eq. (\ref{defpartial}) it follows that
\begin{eqnarray}
\frac{\partial \mathcal{L}_i}{\partial \mathbf{x}_i}=\frac{\partial \mathcal{I}_{ij}^{+}}{\partial \mathbf{x}_i} + \frac{\partial \mathcal{I}_{ij}^{-}}{\partial \mathbf{x}_i},
\end{eqnarray}
which is a linear function of the Lebesgue integrable accelerations derived from the velocities of bounded variation.
\par
Using results (a) and (b) of above allows the use of Lebesgue dominated convergence to place the limit inside the integral and the Gateaux derivative  $\delta S_1$ can be expressed as
\begin{eqnarray}
\delta S_1 = \int_{t_{O_{1}}}^{t_{L_2^{-}}} \left[   \frac{\partial \mathcal{L}_1 }{\partial \mathbf{x}_1}\cdot \mathrm{b}_{1}   +   \frac{\partial \mathcal{L}_1}{\partial \dot{\mathbf{x}}_1} \cdot {\dot{\mathrm{b}}}_{1}   \right]dt ,
\label{Frechet}
\end{eqnarray}
where again partial derivatives are evaluated with the Euclidean norm of $ \mathbb{R}^3$ for both $\mathbf{x}_1$ and $\dot{\mathbf{x}}_1$ and center dot denotes the scalar product of $\mathbb{R}^3$. Notice that both integrals on the right-hand-side of Eq. (\ref{Frechet}) exist and are bounded in $X_{BV}$. Notice that, unlike the integrals in (\ref{aFokker}), all integrals in (\ref{Frechet}) extend over the same range $[t_{O_1},t_{L_2^-}]$ because the boundary segments of the trajectories are kept fixed for the trajectory variations of $\delta S_1$. Likewise, the $\delta S_2$ variation has a form analogous to (\ref{Frechet}) involving the partial derivatives of the partial Lagrangian $\mathcal{L}_2(\mathbf{x}_2, \dot{\mathbf{x}}_2,t)$ integrated over the range $[t_{O_1^+},t_{L_2}]$, as given by Eqs. (\ref{defpartial}) and (\ref{partial}) with $i=2$.
\par
 A velocity discontinuity at a point of trajectory $i$ naturally creates a velocity discontinuity at either the points on the future or the past light-cones (\ref{light-cone}) and along trajectory $j$, which motivates the use of the following \textit{sewing chain} of grid times. In $X_{\widehat{C}^2}$ we define a finite number of ``grid times" along each orbit by the following forward sewing chain procedure. As illustrated in Fig. 1, we start from a point on the past history of particle $2$ and move to the corresponding point on its forward light-cone and along trajectory $1$, then to the respective point on the forward light-cone of the later point and along trajectory $2$ and so on until a last point on the future history of particle $1$. This defines times $t_{O_1} \equiv \tau_{0}^{(1)}<\tau_1^{(1)}<...<\tau_{N_1+1}^{(1)} \equiv t_{L_2^-}$ along  trajectory $1$ and times $t_{O_1^{+}} \equiv \tau_{0}^{(2)}<\tau_1^{(2)}<...<\tau_{N_1+1}^{(2)} \equiv t_{L_{2}}$ along trajectory $2$. We further use perturbations $(\mathrm{b}_1(t), \mathrm{b}_2(t)) \in X_{BV}$ possessing derivative discontinuities only along the former grid times, whose precise locations actually depend on the sizes of $(\mathrm{b}_1(t), \mathrm{b}_2(t)) $.  In such a setup  we can perform piecewise integrations by parts in the expressions of the $\delta S_i$ because both $\frac{d}{dt}(\frac{\partial \mathcal{L}_i}{\partial \dot{\mathbf{x}_i}}) $ are piecewise continuous between grid times for $i=1,2$. Notice that due to state dependency, the times $\tau_{\alpha}^{(i)}$ along each finite grid depend on both perturbations $\mathrm{b}_1(t)$ and $ \mathrm{b}_2(t)$. 

Using the above enlarged grid, the second term of Eq. (\ref{Frechet}) can be integrated by parts piecewise, yielding
\begin{eqnarray}
\delta S_1 &=&\int_{t_{O_1}}^{t_{L_2 ^{-}}}{{\mathrm{b}}}_{1}\cdot \lbrack 
\frac{\partial \mathcal{L}_1 }{\partial \mathbf{x}_{1}}-\frac{d}{dt}(\frac{%
\partial \mathcal{L}_1 }{\partial \dot{\mathrm{x}}_{1}})]dt  \notag \\
&&-\sum\limits_{\sigma =1}^{\sigma =N_1}{\mathrm{b}}_{1}(\tau_{\sigma
})\cdot \Delta (\frac{\partial \mathcal{L}_1 }{\partial \dot{\mathbf{x}}_{1}})|_{\tau_{\sigma}},%
\label{EuLa}
\end{eqnarray}%
where 
\begin{equation}
\Delta (\frac{\partial \mathcal{L}_1 }{\partial \dot{\mathbf{x}}_{1}}%
)|_{\tau_{\sigma}} \equiv \frac{\partial \mathcal{L}_1 }{\partial 
\dot{\mathbf{x}}_{1}}(\tau_{\sigma }^{+})-\frac{\partial \mathcal{L}_1 }{\partial 
\dot{\mathbf{x}}_{1}}(\tau_{\sigma }^{-}),  \label{disc}
\end{equation}%
is the discontinuity of the derivative at point $t=\tau_{\sigma}$ for $\sigma=1,2,...,N_1$. 
\par
As discussed in Ref. \cite{cinderela}, the critical-point-conditions in $X_{\widehat{C}^2}$ are (i) the vanishing of the integrand on the right-hand-side of (\ref{EuLa}) for arbitrary $\mathrm{b}_1(t)$,
\begin{eqnarray}
\frac{\partial \mathcal{L}_1 }{\partial \mathbf{x}_{1}}-\frac{d}{dt}(\frac{%
\partial \mathcal{L}_1 }{\partial \dot{\mathbf{x}}_{1}})=0, \label{Euler-Lagrange}
\end{eqnarray}
henceforth the Euler-Lagrange equation on the ${{C}}^2$ segments, and (ii) the Weierstrass-Erdmann corner condition that $P_1(t)$ must be \emph{continuous} at the $N_1$ isolated velocity-discontinuity points $\tau_\sigma$ of an orbit of $X_{\widehat{C}^2}$ \cite{cinderela,Gelfand}, in order to vanish the discrete sum on the right-hand-side of (\ref{EuLa}) for arbitrary $\mathrm{b}_1(\tau_\sigma)$ and $\sigma=1,...,N$. 
\par
To conform with Ref. \cite{cinderela} and general physics literature we define
\begin{eqnarray}
  \negthickspace \negthickspace  \mathbf{A}_j (t,\mathbf{x}_i)& \equiv &  \frac{\dot{{\mathbf{x}}}_{j}^{-} }{2r_{ij}^{-}(1-\mathbf{n}_{ij}^{-}\cdot \dot{\mathbf{x}}_{j}^{-})} +
\frac{\dot{{\mathbf{x}}}_{j}^{+}}{2r_{ij}^{+}(1+\mathbf{n}_{ij}^{+}\cdot \dot{\mathbf{x}}_{j}^{+})}, \label{DEFA} \\
\negthickspace  \negthickspace \negthickspace \negthickspace {U}_j (t,\mathbf{x}_i)& \equiv & \frac{1 }{2r_{ij}^{-}(1-\mathbf{n}_{ij}^{-}\cdot \dot{\mathbf{x}}_{j}^{-})} +
\frac{1}{2r_{ij}^{+}(1+\mathbf{n}_{ij}^{+}\cdot \dot{\mathbf{x}}_{j}^{+})}, \label{DEFU}
\end{eqnarray}
which is first used to express (\ref{defpartial}) as 
\begin{equation}
\mathcal{L}_i (\mathbf{x}_i, \dot{\mathbf{x}}_i,t)= \mathcal{M}_i(\dot{\mathbf{x}}_i)-\dot{\mathbf{x}}_i \cdot \mathbf{A}_j (t,\mathbf{x}_i)+ U_j(t,\mathbf{x}_i). \label{partial}
\end{equation} 

\section{Stieltjes integration by parts in $X_{BV}$}
\label{Stieltjes}

For the variational problem in $X_{BV}$ one must still see if $P_i(t)$ should be continuous, as obtained in Section \ref{sectionXC2} for $X_{{\widehat{C}^2}  } $. The difficulty is that discontinuity points in $X_{BV}$ are only countable and so could accumulate, invalidating the piecewise integration by parts and Riemann integration used to obtain (\ref{EuLa}) here and in Ref. \cite{cinderela}. The space $X_{BV}$ is henceforth equipped with the norm  
\begin{eqnarray}
\Vert \mathrm{b}_1 \Vert_{BV} \equiv \vert \dot{\mathrm{b}}_1 (t_{O_1})\vert + \mathbb{TV} [ t_{O_1}, t_{L_2^-}, \dot{\mathrm{b}}_1(t) ] , \label{norma}
\end{eqnarray}
where the \emph{total variation} ($\mathbb{TV}$) of $ \dot{\mathrm{b}}_1 (t)$ on  $[t_{O_1},t_{L_2^-}]$ is defined by
\begin{eqnarray}
\mathbb{TV} [ t_{O_1}, t_{L^{-}_{2}}, \dot{\mathrm{b}}_1(t)] \equiv \sup_{P \in \wp} \sum_{i=0}^{N} \vert \dot{\mathrm{b}}
_{1}(\tau_{i+1})-\dot{\mathrm{b}}_{1}(\tau_i) \vert, \label{total_variation}
\end{eqnarray}
with $\wp$ the set of all partitions of $[t_{O_1},t_{L_2^-}]$ in disjoint intervals, as discussed for example in \cite{Ambrosio}.

\begin{lem} \label{LBV}  Along the trajectories of $X_{BV}$,  the momentum function defined by Eq. (\ref{P1}), $P_1(t) \equiv{\partial \mathcal{L}_1 }/{ \partial \dot{{\mathbf{x}}}_{1}}$, is a function of bounded variation.
\label{lema1}
\end{lem}
\begin{proof}
Because the orbits of $X_{BV}$ are sub-luminal, the three denominators of Eq. (\ref{momentum1}) are bounded away from zero.
Formula (\ref{momentum1}) for the momentum involves well-defined algebraic operations between functions that are either absolutely continuous or at the most of bounded variation, yielding a function of bounded variation because of the Banach-algebra property of the BV space\cite{Ambrosio}. \end{proof}

\begin{lem}  Under the assumptions of Lemma \ref{LBV}, the second term of the right-hand-side of (\ref{Frechet}) becomes a Stieltjes integral, $\int_{t_{O_{1}}}^{t_{L_2^{-}}} P_1(t) \cdot \dot{\mathrm{b}}_1 (t) dt=-\int \mathrm{b}_1 \cdot dP_1$.
 \label{la2}
\end{lem}
\begin{proof}
The first part of the proof rests on the fact that $\mathrm{b}_1(t)$ is absolutely continuous and on Lemma \ref{lema1} assuring that $P_1(t)$ is of bounded variation, such that 
\begin{eqnarray}
\int_{t_{O_{1}}}^{t_{L_2^{-}}}   P_1(t) \cdot \dot{\mathrm{b}}_{1}(t)   dt =\int P_1 \cdot d\mathrm{b}_1. \label{eq1lema2}
\end{eqnarray}
Second, since $P_1(t)$  is a function of bounded variation under the assumptions of Lemma \ref{lema1}, and $\mathrm{b}_1(t)$ is absolutely continuous in $X_{BV}$, we can perform Stieltjes integration by parts \cite{BelaNagy, Folland} on the right-hand-side of Eq. (\ref{eq1lema2}), yielding
\begin{eqnarray}
\int P_1 \cdot d\mathrm{b}_1=[ P_1(t_{L_2^{-}}) \cdot \mathrm{b}_1(t_{L_2^{-}})-P_1(t_{O_1}) \cdot \mathrm{b}_1(t_{O_1})] -\int \mathrm{b}_1 \cdot dP_1. 
\label{eq2lema2}
\end{eqnarray}
Finally, the boundary term on the right-hand-side of Eq. (\ref{eq2lema2}) vanishes due to the boundary conditions (\ref{boundaries}), finishing the proof.
\label{lema2}
\end{proof}
\par
To study the critical-point-condition in $X_{BV}$ we start by expressing the second term on the right-hand-side of (\ref{Frechet}) using the integral identity of Lemma \ref{la2}, 
\begin{eqnarray}
\delta S_1 = \int_{t_{O_{1}}}^{t_{L_2^{-}}}  (\mathrm{b}_{1} \cdot \frac{\partial \mathcal{L}_1}{\partial \mathbf{x}_1})dt  \; -\int \mathrm{b}_1 \cdot dP_1.
\label{Frechet-Stieltjes}
\end{eqnarray}
In Eq. (\ref{Frechet-Stieltjes}), $P_1(t)$ is a function of bounded variation by Lemma \ref{LBV} and therefore it has a unique decomposition as a sum of three terms,  $P_1(t)=P_1^{ac}(t)+P_1^{sc}(t)+P_1 ^{J}(t)$, where (i) $P_1^{ac}(t)$ is the absolutely continuous part of $P_1 (t)$, (ii) $P_1^{sc}(t)$ is the singular-continuous part  of $P_1 (t)$, and (iii) $P_1 ^{J}(t)$ is a  jump function containing the denumerable set of jump discontinuities of $P_1(t)$(see for example Ref. \cite{Leoni}). Such a decomposition induces a decomposition of the Borel-Stieltjes measure on the second term of the right-hand-side of (\ref{Frechet-Stieltjes}) into three corresponding measures. The part with the absolutely continuous measure can be integrated back by parts and the jump part yields a sum of jump discontinuities, resulting in a generalized form of (\ref{EuLa}), 
\begin{eqnarray}
\delta S_1 &=&\int_{t_{O_1}}^{t_{L_2 ^{-}}}{{\mathrm{b}}}_{1}\cdot \lbrack 
\frac{\partial \mathcal{L}_1 }{\partial \mathbf{x}_{1}}-\frac{dP_1^{ac}}{dt}]dt  -\int \mathrm{b}_1 \cdot dP_1^{sc}\notag \\
&&-\sum\limits_{k =1}^{k=\infty} {\mathrm{b}}_{1}(t_k)\cdot \Delta P_1^{J}(t_k).%
\label{EuLaGen}
\end{eqnarray}%
On the right-hand-side of Eq. (\ref{EuLaGen}) we have that (a) the derivative $dP_1^{ac} /dt $ of the absolutely continuous part is defined everywhere outside a set of Lebesgue measure zero, making the first Lebesgue integral well-defined for $\mathrm{b}_1(t) \in X_{BV}$ and (b) the Borel-Stieltjes measure on the second integral is concentrated on the singular set of $dP_1^{sc}$, which is an uncountable set of Lebesgue measure zero. For a critical point, the right-hand-side of (\ref{EuLaGen}) must vanish for arbitrary $\mathrm{b}_1(t) \in X_{BV}$. Using a sequence of functions $\{ \mathrm{b}_{1\alpha}^{k}(t)\} $ of increasingly small Lebesgue measure \emph{and} concentrated on each discrete discontinuity point of $P_1^{J}(t)$, the limiting value of the right-hand-side of (\ref{EuLaGen}) depends only on the \emph{value} of $\mathrm{b}_{1\alpha}^{k}(t)$ \emph{at} the discontinuity point $t=t_{k}$. The former implies that the coefficient of $\mathrm{b}_{1\alpha}^{k}(t)| _{t=t_k}$ must vanish for each $t_k$, which in turn implies that $P_1^{J}=0$. Therefore we must have that (i) $ P_1(t)$ is \emph{continuous} and therefore the last term of (\ref{EuLaGen}) vanishes.  Finally,
for a sequence $\{ \mathrm{b}_{1\beta}(t)\} $ of increasingly small Lebesgue measure and otherwise arbitrary functions $\mathrm{b}_{1\beta}(t) \in X_{BV} $, the first term on the right-hand-side of Eq. (\ref{EuLaGen}) has a vanishing limit, and in order to vanish the remaining Stieltjes integral of (\ref{EuLaGen}) we must have $P^{sc}=0$, implying (ii) that $P_1(t)$ must be \emph{absolutely continuous}. The resulting absolutely continuous momentum $P_1^{ac}(t)$ must satisfy the Euler-Lagrange equation (\ref{Euler-Lagrange}) Lebesgue-almost-everywhere in order to vanish the first term on the right-hand-side of (\ref{EuLaGen}) with arbitrary $ \mathrm{b}_1 (t)  \in X_{BV}$.

\section{Discussions}
\label{Discussions}
\subsection{Summary of results}

We have generalized the domain of the electromagnetic functional to a local space $X_{BV}$ of absolutely continuous trajectories possessing velocities of bounded variation and shown that the Gateaux derivative of the electromagnetic functional defines \emph{two} partial Lagrangians, one for each particle. The critical-point conditions in  $X_{BV}$ are Euler-Lagrange equations holding Lebesgue-almost-everywhere \textit{plus} the absolutely continuity conditions for the Legendre transforms and the momentum vectors (\ref{momentum1}) of the \emph{two} partial Lagrangians. While the condition of an Euler-Lagrange equation holding Lebesgue-almost-everywhere is weaker than Euler-Lagrange equations everywhere, the generalized Weierstrass-Erdmann conditions became the stronger conditions of \emph{absolute continuity}. 
\par
Besides the importance to electrodynamics, another reason to study other variational structures came after the no-interaction theorem \cite{no-interaction} exposed the severe limitations of the classical mechanical form $S \equiv \int_{0}^{T} \mathcal{L} (\mathbf{x},\dot{\mathbf x}, t)dt $ to describe Lorentz-invariant dynamics. As far as we know, this is the first time variational principles with \emph{two} partial Lagrangians are studied in functional-analytic detail.
 \subsection{Solution based approaches to the electromagnetic problem}
  The earliest numerical methods for the electromagnetic two-body problem were integration schemes adapted from ordinary differential equations (ODE), and are discussed in Ref.\cite{Igor}. These former methods are unsuitable to the variational boundary-value problem and the early attempts of Ref. \cite{Igor} to apply ODE extrapolation schemes to orbits in the neighbourhood of circular orbits are limited by severe instabilities. Subsequently, the present author and collaborators have devised several numerical methods for the electromagnetic two-body problem: \textbf{(i)} The existence results of Refs. \cite{Driver_C,Driver_D} for the repulsive problem motivated our early numerical study \cite{EFY1} of the Wheeler-Feynman equations with an iterative steepest-descent method combined with optimization to find a self-consistent $C^\infty$ solution. The numerical work with initial condition restricted to a line becomes particularly simple because the advanced and the retarded accelerations do not appear in the equations of motion and the accelerations go to zero asymptotically as the particle separations and the delays become large enough. The numerical studies confirmed  the prediction of Refs. \cite{Driver_C,Driver_D} that the solution \emph{restricted} to a space of globally-defined $C^\infty$ orbits is uniquely determined by finite-dimensional data \cite{Driver_C}. \textbf{(ii)} Reference \cite{EFY2} studied the Wheeler-Feynman equations of motion for two equal masses with \textit{opposite} charges again restricted to globally bounded $C^\infty$ collinear orbits. The study used a numerical regularization of the collision and the same optimization of Ref. \cite{EFY1}. As in the former case, the advanced and the retarded accelerations are absent from the equations of motion, a feature of the collinear motion. On the contrary, for opposite charges the accelerations do not vanish asymptotically for globally bounded orbits and moreover the charges fall into each other along globally bounded orbits, at which times both charges reach the speed of light and ``pass through each other'' \cite{EFY2}. The numerically calculated orbits of Ref. \cite{EFY2} using an 18 parameter optimization turn out to be determined uniquely by finite-dimensional data, a reduction analogous  to the predictions of Refs. \cite{Driver_C,Driver_D} for the repulsive case\cite{EFY1}. \textbf{(iii)} Motivated by the regular orbits found in \cite{EFY1, EFY2} and the relative simplicity of the equations for collinear motion, our work of Ref. \cite{Chaos15} studied a $C^\infty$ orbit in the case of arbitrary masses using the Hamilton-Jacobi theory.  Such extension to the case of arbitrary masses has a physical application in the quantization of the electromagnetic two-body problem. \textbf{(iv)} Still for the collinear motion of opposite charges, Ref. \cite{Politi} studied the Wheeler-Feynman equations in a setup including a (constant) external electric field to balance the attraction and create a fixed point and a Hopf bifurcation. The bifurcating periodic orbits are studied in \cite{Politi} by introducing an ad hoc surrogate system with an adjustable period into the analysis, which transforms future data into past data by exploiting the periodicity, thus obtaining a system with delays only. The periodic orbits are studied with the integrator for differential-algebraic equations with state-dependent delay RADAR5. \textbf{(v)} For generic boundary data the two-body problem with advanced and delayed arguments is unsuitable for direct integration as an initial value problem and moreover in three space dimensions the functional differential equation is of neutral type. In Ref. \cite{CAM}  the present author and collaborators studied the attractive two-body problem as a boundary value problem solved numerically with a finite element method with a $C^1$ smooth Galerkian. The method\cite{CAM} still could not deal with velocity discontinuities and had to adjust the past and the future boundary segments to yield $C^1$ solutions. \textbf{(vi)} Finally, an example of the effectiveness of the variational principle when compared to the neutral differential-delay equation formulation is our first numerical method capable of dealing with velocity discontinuities for boundaries having the shortest time separation between boundary segments across the lightcone \cite{Daniel}.  A property of the shortest-length boundaries setup of Ref. \cite{Daniel} is that one particle's retarded position falls in the past boundary segment while the other particle's advanced position falls in the future boundary segment, thus reducing the functional equation to a shooting problem for a non-autonomous ODE. The method \cite{Daniel} starts from boundary segments that include velocity discontinuities and marches an ODE integrator from the initial point to the endpoint, stopping to enforce the Weierstrass-Erdmann conditions at each breaking point. Solutions with discontinuous velocities are exhibited in Ref.\cite{Daniel} for several sets of infinite-dimensional boundary data.
  
  \subsection{Advantages and Limitations}
 
 An advantage of the variational formulation is that one is dealing with a problem of functional minimization \cite{JMP2009,CAM}, unlike the case of breaking points for generic neutral differential-delay equations\cite{BellenZennaro}. In the variational problem the breaking points become Weierstrass-Erdmann corners satisfying the generalized Weierstrass-Erdmann absolute continuity conditions discussed below Eq. (\ref{EuLaGen}). The functional minimization advantage allows the use of optimization methods in the numerical studies \cite{CAM, Daniel, EFY1, EFY2}. A limitation still to be overcome is the construction of a complete normed space domain for the electromagnetic functional, after which the mountain pass theorem and Ekeland's variational principle \cite{Ambrosetti} can be used to study solutions.

\section{Appendix}

In the main text we avoided using the relativistic four-space constructions of \cite{JMP2009} for simplicity. On the other hand, this natural four-space associated with the electromagnetic problem has advantages, allowing a particularly simple derivation of the second Weierstrass-Erdmann condition.  The electromagnetic functional (\ref{aFokker}) has a natural definition in the larger space\cite{JMP2009} which is discussed below to motivate an extension of the domain $X_{BV}$ to \textit{four-space}. In a nutshell, the time variable of 
each trajectory of $X_{BV}$ can be replaced by an arbitrary monotonically increasing  function $t_i(s)$ of another independent variable $s$, and the functional minimization can be extended to include these two extra arbitrary functions.  
 Trajectories of the larger ambient space $ T_{BV} \otimes X_{BV} \equiv\mathbb{R} \times \mathbb{R}^3$ are absolutely continuous functions of the real variable $s \rightarrow  Q_i (s) \equiv (t_i(s), \mathbf{x}_i(s)) $ possessing monotonically increasing time-components $t_i(s)$ for $i=1,2$. The functions $ (t_i(s), \mathbf{x}_i(s)) $ are obtained by integration of the respective derivatives of bounded variation,
\begin{eqnarray}
({t}_i^{\prime}(s),{\mathbf{x}}_i^{\prime}(s)).
 \label{defS}
\end{eqnarray}
In equation (\ref{defS}) and henceforth in this appendix, a prime denotes derivative respect to the real variable $s$, to distinguish it from the dot denoting time-derivative in the main text.  The chain rule holds for differentiation respect to the absolutely continuous time, yielding a derivative defined almost everywhere by 
\begin{eqnarray}
\dot{\mathbf{x}}_i = {\mathbf{x}}_i^{\prime}(s)/{t}_i^{\prime}(s), 
\label{Chain_rule}
\end{eqnarray}
wherever ${\mathbf{x}}_i^{\prime}(s)$ and ${t}_i^{\prime}(s)>0$ are defined.
\par
 The state-dependent light-cone condition (\ref{light-cone}) has a natural extension in the larger space $T_{BV} \otimes X_{BV}$, i.e.
\begin{equation}
t_j(s_j^{\pm})= t_i(s) \pm \vert {\mathbf{x}_i}(s)-{\mathbf{x}}_{j}(s_j^{\pm}) \vert,
\label{light-cone-S}
\end{equation}
as well as the electromagnetic functional, whose formulas generalize from (\ref{aFokker}) and (\ref{L2}) to
\begin{align} \label{aFokkerS}
S \equiv &\negthickspace \int_{s_{O_1^{+}}}^{s_{L_{2}}}\negthickspace {\widetilde{M}}_{2}ds_{2}
+\negthickspace \int_{s_{O_{1}}}^{s_{L_2^{-}}}\negthickspace  \widetilde{M}_{1}ds_{1}+\underbrace{ \int_{s_{O_{1}}}^{s_{L_2^{+}}} \negthickspace \mathbb{I}_{12}^-ds_{1}}
+ \underbrace{\negthickspace \int_{s_{O_1}}^{s_{L_2^-}}\negthickspace  \widetilde{I}_{12}^+ds_{1}},\\
& \qquad \qquad   \qquad\qquad\qquad\qquad \quad \; \, \, \; \;  \;\Updownarrow \qquad\qquad\quad\; \; \Updownarrow \notag\\
\negthickspace \negthickspace =&\negthickspace \int_{s_{O_{1}}}^{s_{L_2^{-}}} \negthickspace \widetilde{M}_{1}ds_{1}
+\negthickspace \int_{s_{O_1^{+}}}^{s_{L_{2}}}\negthickspace \negthickspace \widetilde{M}_{2}ds_{2}
+ \overbrace{ \int_{s_{O_1^-}}^{s_{L_2}} \negthickspace \widetilde{I}_{21}^+ds_{2}}
+ \overbrace{ \int_{s_{O_1^+}}^{s_{L_{2}}} \negthickspace \widetilde{I}_{21}^-ds_{2}}.
\label{L2S}
\end{align}
Again, as with Eqs. (\ref{aFokker}) and (\ref{L2}), the vertical arrows linking an integral of Eq. (\ref{aFokkerS}) to an integral of Eq. (\ref{L2S}) indicate equality under a change of the integration variable using the state-dependent condition (\ref{light-cone-S}). 
The Radon-Nikodym derivative to change from the variable $s_i$ to either  the other particle's advanced or retarded parameter $s_{j\pm}$ is obtained from the implicit condition (\ref{light-cone-S}), 
\begin{equation}
\frac{d s_{j\pm}}{d s_{i}}= \frac{({t}_i^{\prime} \pm \mathbf{n}_{ij}^{\pm} \cdot {{\mathbf{x}}}_i^{\prime})}{({t^{\prime}}_j^{\pm} \pm \mathbf{n}_{ij}^{\pm} \cdot {{\mathbf{x^{\prime}}}}_{j}^{\pm})} ,  \label{t2diffS1}
\end{equation}
 which is defined almost everywhere in $T_{BV} \otimes X_{BV}$. In Eq. (\ref{t2diffS1}) the $\pm$  superscripts indicate that the derivatives in the denominator of Eq. (\ref{t2diffS1}) are to be evaluated at either the advanced or the retarded deviating argument $s_j^{\pm}$ determined by condition (\ref{light-cone-S}), and unit vector $\mathbf{n}_{ij}^{\pm } $ is still defined by Eq. (\ref{unit-n}).  The partial Lagrangians defined by the generalized integrands of (\ref{aFokkerS}) and (\ref{L2S}) with trajectory $i$ kept fixed are 
 \begin{eqnarray}
\widetilde{L}_i(t_i(s),\mathbf{x}_i(s), t_i^{\prime}(s), {\mathbf{x}}_i^{\prime} (s),s) \equiv  \widetilde{M}_i + \widetilde{I}_{ij}^{-} +\widetilde{I}_{ij}^{+} \label{defpartial-S}, 
\end{eqnarray}
where
\begin{eqnarray}
\widetilde{M}_i\equiv m_i ( {t}_i^{\prime}-\sqrt{{t^{\prime}_i}^2-{{\mathbf{x}}_ i^{\prime}}^2 }   \;), \label{MUS}
\qquad \\
\widetilde{I}^\pm_{ij}(Q_i,\dot{Q}_i,{Q}_{j}^{\pm},{\dot{Q}}_{j}^{\pm})
\equiv \frac{({t^{\prime}}_i {t^{\prime}}_{j\pm}-{\mathbf{x^{\prime}}}_i \cdot {\mathbf{x^{\prime}}}_{j}^{\pm}) }{2r_{ij\pm}( {t^{\prime}}_{j}^{\pm} \pm \mathbf{n}_{ij}^{\pm}\cdot {\mathbf{x^{\prime}}}_{j}^{\pm})  }.  \label{EUS}
\end{eqnarray}
This surprising extension is due to the special form of the electromagnetic functional, for which the tilde functions (\ref{MUS}) and (\ref{EUS}) can be expressed using (\ref{MU}) and (\ref{EU}) with an extra linear scaling by $t_i^{\prime}(s)$, i.e.,
\begin{eqnarray}
\widetilde{L}_i(t_i(s),\mathbf{x}_i(s), t_i^{\prime}(s), {\mathbf{x}}_i^{\prime} (s),s) = 
(\mathcal{M}_i + \mathcal{I}_{ij}^{-} +\mathcal{I}_{ij}^{+})t_i^{\prime}(s) \notag \\ =t_i^{\prime}(s)\mathcal{L}_i (\mathbf{x}_i(t_i(s)), \dot{\mathbf{x}}_i(t_i(s)),t_i(s)), \label{homogeneous}
\end{eqnarray}
precisely the necessary Radon-Nikodym derivative for the change of variables to reduce the integrals (\ref{aFokkerS}) and (\ref{L2S}) respectively to integrals (\ref{aFokker}) and (\ref{L2}). Extension to $T_{BV} \otimes X_{BV}$ is possible because of two facts: (i) the linear dependence on the first derivatives of the electromagnetic functional (\ref{aFokkerS}) and (\ref{L2S}), and (ii) the only other dependence on derivatives is in the square root definition (\ref{MUS}) of $\widetilde{M}_i$, which is a homogeneous function of $t_i^{\prime}(s)$. The above formulas reduce to the respective $X_{BV}$ formulas used in the main text for the special case when $ t_i (s) \equiv s $ for $i=1,2$.

\par

\begin{thm} \label{theorem3}  The Legendre transforms of the partial Lagrangians (\ref{partial}) are functions of bounded variation in $T_{BV} \otimes X_{BV}$, in addition to the momentum of each partial Lagrangian, which are the same defined by Eq. (\ref{P1}) with $i=1,2$.

\begin{proof} The proof is a generalization of the proof given below Eq. (\ref{EuLaGen}). The perturbed orbits in $T_{BV} \otimes X_{BV}$ are Eqs. (\ref{perturb}), (\ref{perturb2}), (\ref{boundaries}) and (\ref{boundaries2}) with two extra monotonic time-components, that we henceforth indicate using sans-serif font to differentiate it from the main text, i.e., $\mathsf{b}_i (s) \equiv(t_i(s), \mathrm{b}_i(s))$ and $\mathsf{x}_i (s) \equiv (t_i(s), \mathbf{x}_i(s))$ for $i=1,2$. The generalization of (\ref{EuLaGen}) to $T_{BV} \otimes X_{BV}$ is

\begin{eqnarray}
\delta \widetilde{S}_i &=&\int_{s_i^0}^{s_i^f} {{\mathsf{b}}}_{i}\cdot \lbrack 
\frac{\partial \widetilde{L}_i }{\partial \mathsf{x}_{i}}-\frac{d\widetilde{P}_i^{ac}}{ds}]ds  -\int \mathsf{b}_i \cdot d\widetilde{P}_i^{sc}\notag \\
&&-\sum\limits_{k =1}^{k=\infty} \mathsf{b}_{i}(t_k)\cdot \Delta \widetilde{P}_i^{J}(t_k),%
\label{EuLaGenLarge}
\end{eqnarray}
where $(s_1^0,  s_1^f)\equiv (s_{O_1},s_{L_2^-})$ and $(s_2^0, s_2^f) \equiv ( s_{O_1^+} , s_{L_2})$.  In Eq. (\ref{EuLaGenLarge}) the center dot denotes the scalar product of $\mathbb{R} \times \mathbb{R}^3 \equiv \mathbb{R}^4$  and the partial derivative is taken with respect to the Euclidean norm of the same $\mathbb{R}^4$.
The momentum of $T_{BV} \otimes X_{BV}$ appearing in  Eq. (\ref{EuLaGenLarge}) is $\widetilde{P}_i \equiv \frac{\partial \widetilde{L}_i }{\partial  {\mathsf{x}}_i^{\prime}}(\mathsf{x}_i(s), {\mathsf{x}}_i^{\prime} (s),s)$, which is a function of bounded variation by the arguments below Eq. (\ref{EuLaGen}). 
The monotonically increasing  time component of $T_{BV} \otimes X_{BV}$ induces an extra, zeroth component for each particle's momentum (\ref{momentum1}), namely
\begin{eqnarray}
\widetilde{P}_i^{0} \equiv \frac{\partial \widetilde{L}_i }{ \partial t_i^{\prime}} = m_i-\frac {m_i{t^{\prime}_i}} {\sqrt{{t^{\prime}_i}^2-{{\mathbf{x}}_ i^{\prime}}^2 }}  + \frac{ {t^{\prime}}_{j-} }{2r_{ij-}( {t^{\prime}}_{j}^{-} -\mathbf{n}_{ij}^{-}\cdot {\mathbf{x^{\prime}}}_{j}^{-})  }+ \frac{ {t^{\prime}}_{j+} }{2r_{ij+}( {t^{\prime}}_{j}^{+} + \mathbf{n}_{ij}^{+}\cdot {\mathbf{x^{\prime}}}_{j}^{+})  }.  
\label{WE2S}
\end{eqnarray}
 Dividing the numerator and the denominator of each of the last three terms of Eq. (\ref{WE2S}) by either $t^{\prime}_i \in BV$ or $t^{\prime}_{j\pm} \in BV$  and using the chain rule $\dot{\mathbf{x}}_i = {\mathbf{x}}_i^{\prime}/{t}_i^{\prime}$ to pass from $T_{BV} \otimes X_{BV}$ to $X_{BV}$ yield
\begin{eqnarray}
\widetilde{P}_i^{0} \equiv  m_i - \frac {m_i}{\sqrt{1 - \dot{\mathbf{x}}_i^2}}  + U_j(t,\mathbf{x}_i),
\label{WE2T}
\end{eqnarray}
where $U_j(t,\mathbf{x}_i)$ is defined by (\ref{DEFU}).
A direct calculation shows that $\widetilde{P}_i^{0}$ defined by (\ref{WE2T}) is minus the Legendre transform of the partial Lagrangian $\mathcal{L}_i$ of Eq. (\ref{partial}), 
\begin{eqnarray}
 \widetilde{P}_i^{0}= - ( \dot{\mathbf{x}}_i \cdot \frac {\partial \mathcal{L}_i}{\partial \mathbf{x}_i} - \mathcal{L}_i ) . \label{Legendre_equality}
 \end{eqnarray}
Once the Legendre transform of $\mathcal{L}_i(\mathbf{x}_i, \dot{\mathbf{x}}_i, t)$ is equal to minus the momentum of $T_{BV} \otimes X_{BV}$, it must be an absolutely continuous function in $T_{BV} \otimes X_{BV}$ by the argument below Eq. (\ref{EuLaGen}). Since $t_i(s)$ is monotonically increasing and therefore an invertible function, (\ref{WE2T}) must be an absolutely continuous function of $t$ as well, proving the first part of the statement.
\par
The momentum components in $T_{BV} \otimes X_{BV}$ are 
\begin{eqnarray}
\widetilde{P}_{\mathbf{x}_i} (t) \equiv   \frac{\partial \widetilde{L}_i}{\partial \mathbf{x}^{\prime}_i }   =\frac{m_{i} \mathbf{x}^{\prime}_{i}}  { \sqrt{ {t^{\prime}_i}^2 - {\mathbf{x}_i^{\prime}}^2  }   }
-\frac{    {{\mathbf{x}}}_{j}^{\prime -}      }{    2r_{ij}^{-}(1-\mathbf{n}_{ij}^{-}\cdot {\mathbf{x}}_{j}^{\prime -})          } 
-\frac{     {{\mathbf{x}}}_{j}^{\prime +}}{2r_{ij}^{+}(1+\mathbf{n}_{ij}^{+}\cdot {\mathbf{x}}_{j}^{\prime +})     }.
\label{momentumtilde} 
\end{eqnarray} 
Again, dividing the numerator and the denominator of the last three fractions of Eq. (\ref{momentumtilde}) by either $t^{\prime}_i \in BV$ or $t^{\prime}_{j\pm} \in BV$  and using the chain rule $\dot{\mathbf{x}}_i = {\mathbf{x}}_i^{\prime}/{t}_i^{\prime}$ to pass from $T_{BV} \otimes X_{BV}$ to $X_{BV}$ yield the same momentum of $X_{BV}$, i.e., Eq.  (\ref{momentum1}), completing the proof. 
\end{proof}
\end{thm}
\par
The last implication of  (\ref{EuLaGenLarge}) is that the generalized momenta must satisfy the Euler-Lagrange equations almost-everywhere, 
\begin{eqnarray}
\frac{\partial \widetilde{L}_i }{\partial \mathsf{x}_{i}}-\frac{d}{ds}(\frac{%
\partial \widetilde{L}_i }{\partial \mathsf{x}^{\prime}_i })=0, \label{general_Euler-Lagrange}
\end{eqnarray}
where again the partial derivatives are taken with respect to the Euclidean norm of $\mathbb{R}^4$ for both $\mathsf{x}$ and $\mathsf{x}^{\prime}$.  The scaling property (\ref{homogeneous}) implies that 
\begin{eqnarray}
\frac{\partial \widetilde{L_i}}{\partial \mathsf{x}_i}= t^{\prime}_i\frac{\partial \mathcal{L}_i}{\partial \mathsf{x}_i},  \label{scaling}
\end{eqnarray}
for the four components of Eq. (\ref{general_Euler-Lagrange}) in  $T_{BV} \otimes X_{BV}$. For the spatial components of Eq. (\ref{general_Euler-Lagrange}) Theorem \ref{theorem3} gives that $ \frac{\partial \widetilde{L}_i}{\partial \mathbf{x}^{\prime}_i } = \frac{\partial \mathcal{L}_i}{\partial \dot{\mathbf{x}}}$, which substituted into (\ref{general_Euler-Lagrange}) yields
\begin{eqnarray}
\frac{\partial \mathcal{L}_i }{\partial \mathbf{x}_{i}}-\frac{d}{t^{\prime}_i ds}(\frac{%
\partial \mathcal{L}_i }{\partial \dot{\mathbf{x}}_i })=\frac{\partial \mathcal{L}_i }{\partial \mathbf{x}_{i}}-\frac{d}{dt}(\frac{%
\partial \mathcal{L}_i }{\partial \dot{\mathbf{x}}_i })=0, \label{spatial_Euler-Lagrange}
\end{eqnarray}
where we have used $t^{\prime}_i$ as the Radon-Nikodym derivative to pass from $t^{\prime}_i ds$ to $dt$.
Equation (\ref{spatial_Euler-Lagrange}) is the same Euler-Lagrange equation (\ref{Euler-Lagrange}) of $X_{BV}$.
\par
Finally, we show that the time component of (\ref{general_Euler-Lagrange}) is not a new condition. Substituting Eqs. (\ref{Legendre_equality}) and (\ref{scaling}) into the time component of Eq. (\ref{general_Euler-Lagrange}) yields
\begin{eqnarray}
t_i^{\prime} \frac{\partial L_i}{\partial t}-\frac{d}{ds}(\mathcal{L}_i-\dot{\mathbf{x}_i} \cdot \frac {\partial \mathcal{L}_i}{\partial \dot{\mathbf{x}}_i})=0. \label{fourth}
\end{eqnarray}
Using $t^{\prime}_i$ as a Radon-Nikodym derivative to change the total derivative of (\ref{fourth}) to $\frac{d}{ds}=t_i^{\prime}  \frac {d}{dt}$ and dividing out the non-zero $t_i^{\prime}$ factor transform Eq. (\ref{fourth}) into
\begin{eqnarray}
 \frac{\partial L_i}{\partial t}-\frac{d}{dt}(\mathcal{L}_i-\dot{\mathbf{x}_i} \cdot \frac {\partial \mathcal{L}_i}{\partial \dot{\mathbf{x}}_i})=0,
\end{eqnarray}
which is identically zero because $\mathcal{L}_i(\mathbf{x}_i, \dot{\mathbf{x}_i},t)$ satisfies (\ref{Euler-Lagrange}).

\section{Acknowledgements}

This work was partially supported by the FAPESP grant 2011/18343-6.

\end{document}